\documentclass[10pt, final, journal, letterpaper, twocolumn]{IEEEtran}
\makeatletter
\def\ps@headings{%
\def\@oddhead{\mbox{}\scriptsize\rightmark \hfil \thepage}%
\def\@evenhead{\scriptsize\thepage \hfil \leftmark\mbox{}}%
\def\@oddfoot{}%
\def\@evenfoot{}}
\makeatother 

\IEEEoverridecommandlockouts

\usepackage{amsfonts}
\usepackage[dvips]{graphicx}
\usepackage{caption}
\usepackage{subfigure}
\usepackage{times}
\usepackage{cite}
\usepackage{lettrine}
\usepackage{amsmath}
\usepackage{amsmath}
\allowdisplaybreaks[4]
\usepackage{array}
\usepackage{amssymb}

\usepackage{stfloats}
\usepackage{slashbox}
\usepackage{graphicx}
\usepackage{footnote}
\usepackage{booktabs}
\usepackage{array}
\usepackage{algorithmic}
\usepackage{algorithm}
\usepackage{subeqnarray}
\usepackage{cases}
\usepackage{threeparttable}
\usepackage{color}
\DeclareMathOperator*{\argmax}{argmax}

\newtheorem{proposition}{\underline{Proposition}}

\begin{document}
\bibliographystyle{IEEEtran}

\title{Power Allocation for Secure SWIPT Systems with Wireless-Powered Cooperative Jamming}
\author{Mengyu Liu and Yuan Liu,~\IEEEmembership{Member,~IEEE}

\thanks{

The authors are  with the School of Electronic and Information Engineering,
South China University of Technology, Guangzhou 510641, China (e-mail: liu.mengyu@mail.scut.edu.cn,  eeyliu@scut.edu.cn).
}
}
\maketitle

\vspace{-1.5cm}

\begin{abstract}
This paper considers wireless-powered cooperative  jamming (CJ)  to secure communication between a transmitter (Tx) and an information receiver (IR), in the presence of an energy receiver (ER) which is termed as a potential eavesdropper.
%
%
The full-duplex jammer harvests energy from the Tx's information signal and transmits jamming signal at the same time, where the jamming signal not only confounds the ER (potential eavesdropper) but also charges the ER.
%
Our goal is to maximize the secrecy information rate by jointly optimizing the power allocation at the Tx and jammer while maintaining the harvested energy requirement of the ER.
The studied problem is  non-convex and we propose the optimal solution based on the Lagrange method.
Simulation results show that the proposed scheme significantly outperforms the benchmark schemes.

\end{abstract}


\section{Introduction}

Recently, physical layer security has been investigated extensively  to secure wireless communications.
%
%
For the main methods of physical layer security, artificial noise (AN) and cooperative jamming (CJ) are very promising.
For the former case, the AN signal is transmitted into the null space of the desired signal to degrade the wiretap channel \cite{AN_liu_fading}.
While for the later the external jammer transmits the jamming signal to combat against eavesdropping \cite{relay1}.

On the other hand,  wireless information and power transfer becomes   an appealing  solution to prolong the lifetime of energy-constraint nodes. However, the energy receivers (ERs) are usually deployed relatively closer to the transmitter (Tx), thus the information receivers (IRs) are easily eavesdropped by the ERs.
%
%
A handful of works have considered the physical layer security by wireless energy transfer \cite{liuyuan,hybrid_BS,harvest_and_jam,wireless_powered_friendly_jammer,accumulate_FD}.
For instance, in \cite{hybrid_BS}, the hybrid base station first charges the energy-free source and then  performs CJ  when the source transmits information to the multiple destinations.
In \cite{harvest_and_jam}, multiple wireless-powered jammers were used to secure two-hop relay networks by designing the beamforming matrices. The authors in \cite{wireless_powered_friendly_jammer} conducted wireless power transfer for the jammer and analyzed the throughput.
An ``accumulate-and-jam" protocol was proposed in \cite{accumulate_FD} where the jammer was powered by the source and  secrecy performance metrics were investigated by Markov chain.
Note that the above works require a dedicated energy signal to  power the jammer.

In this paper, we consider the secrecy communication in an orthogonal frequency division multiplexing  (OFDM) based simultaneous wireless information and power transfer (SWIPT) system, which consists of one Tx, one IR, one ER (potential eavesdropper) and one friendly jammer  as shown in Fig. \ref{fig:ff1}. We assume that the Tx has constant energy and  the jammer has no embedded power supply thus needs to harvest energy from the Tx.
The Tx, IR and ER are  equipped with a single-antenna, while the friendly jammer  has two antennas, one for harvesting energy  from the Tx's information signal to the IR  and the other for transmitting jamming signals  to the ER simultaneously  by the full-duplex  capability.
By assuming that the jamming signal can be cancelled at the IR but cannot be removed at the ER (potential eavesdropper), we jointly optimize the transmit power of the Tx and jammer over subcarriers (SCs) to maximize the secrecy rate of the IR while satisfying the energy requirement of the ER. Optimal solution is derived to solve the non-convex optimization problem. Simulation results show that the huge superiority of the proposed method over conventional schemes.


Compared with the works of wireless-powered CJ \cite{hybrid_BS,harvest_and_jam,wireless_powered_friendly_jammer,accumulate_FD}, the differences of our paper are three-fold:
1) We consider a secure SWIPT system where a full-duplex jammer is wireless-powered  by the information signals sent to the IR. There is no need for dedicated energy signal as in other related works;  2) By adopting the  cancellation mechanism of jamming signal at the IR, the secrecy performance of the system can be greatly enhanced; 3) Optimal power allocation of the Tx and jammer are adapted
over SCs to explore frequency flexibility.

\begin{figure}[t]
\begin{centering}
\centering
\includegraphics[scale=0.5]{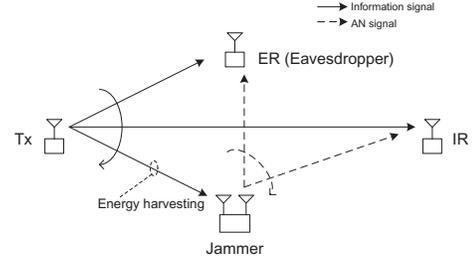}
\vspace{-0.1cm}
\caption{System model of the self-powered secrecy SWIPT.}\label{fig:ff1}
\end{centering}
\vspace{-0.1cm}
\end{figure}

\section{System Model and Problem Formulation}

Consider an OFDM-based  secrecy  SWIPT network consisting a Tx, an IR, an ER (potential eavesdropper) and a jammer as shown in Fig. \ref{fig:ff1}. The Tx, IR, and ER are equipped  with a single-antenna, while the jammer is equipped with two antennas with one for energy harvesting and the other for transmitting jamming signal. That is, when  the Tx transmits information-bearing signal to the IR, the ER harvests energy and may intercept the information. Meanwhile, the jammer uses one antenna to harvest energy from the same information-bearing signal and the other antenna to transmit jamming signal simultaneously by using the harvested energy, thanks to the full-duplex capability. Here it is assumed perfect isolation between the two antennas at the jammer such that self-interference cancellation is perfect.
Note that a small time lag is needed at the initial frame for such a full-duplex system,  which can be negligible if the whole duration of the transmission frame is long enough.
It is worth noting that the jamming signal from the jammer not only jams the ER (potential eavesdropper) but also acts as a source to power the ER.
We assume that the system has $N$ SCs.
The channel power gains  on SC $n$ from the Tx to IR, ER and jammer are denoted as $h_{I,n}, h_{E,n}$ and $h_{J,n}$, respectively, and  the channel power gains from the jammer to IR and ER are denoted as $g_{I,n}$ and $g_{E,n}$.

Assume that the total transmit power of the Tx is $P$ and the power transmitted by the Tx and jammer on SC $n$ are denoted as $p_n$ and $q_n$, respectively. In addition, we consider that there is a peak power constraint on $p_n$ and $q_n$, i.e., $0\leq p_n \leq \bar{p}$, $0 \leq q_n \leq \bar{q}$ for $n=1,\cdots,N$. The total transmit power constraint at the Tx can be given by
\begin{align}
\sum\limits_{n=1}^N p_n \leq P. \label{total_P}
\end{align}
As the energy-free jammer is  powered by the information signal sent from the Tx, the transmit power $q_n$  of  the jammer is constrained by
\begin{align}
\sum\limits_{n=1}^N q_n \leq \zeta\sum\limits_{n=1}^N p_n h_{J,n},\label{jammer_EH_C}
\end{align}
where $\zeta$ is energy conversion efficiency and we assume that $\zeta=1$ for convenience in the next.

The harvested power at the ER comprises of two components with one from the Tx and the other from the jammer, which should satisfy the minimum requirement $\overline{Q}$:
\begin{eqnarray} \label{Q}
\sum\limits_{n=1}^N(p_nh_{E,n}+q_ng_{E,n})\geq \overline{Q}.
\end{eqnarray}

We assume that the jamming signal transmitted by the jammer can be cancelled at the IR but cannot be removed at the ER. This can be practically justified by the similar method in \cite{zhangmeng}: A large set of random sequences (jamming signals) with Gaussian distribution are pre-stored at the jammer and their indices are the keys.
The jammer randomly selects a sequence (jamming signal) and sends its key to the IR over each SC $n$.  The key can be sent in a secret manner via channel independence and reciprocity. As the random sequence is only known at the IR, any potential eavesdropper cannot access the random sequence at each SC.
With this scheme, the achievable information rate of the IR and ER on SC $n$ can be respectively given by
\begin{align} \label{IR_cancel}
r_n = \log_2\bigg(1 + \frac{p_nh_{I,n}}{\sigma^2}\bigg),
r_n^e = \log_2\bigg(1 + \frac{p_nh_{E,n}}{\sigma^2+q_ng_{E,n}}\bigg).
\end{align}

Then  the secrecy  rate on SC $n$ is given by
%
%
%
\begin{align} \label{rate_express}
R_n=[r_n - r_n^e]^+
=\left\{
   \begin{array}{ll}
     r_n - r_n^e, & \hbox{if $q_n \geq A_n$} \\
     0, & \hbox{otherwise,}
   \end{array}
 \right.
\end{align}
where $[\cdot]^+\triangleq\max(0,\cdot)$ and $A_n \triangleq \bigg[\frac{\sigma^2(h_{E,n}-h_{I,n})}{h_{I,n}g_{E,n}}\bigg]^+$.

We consider the instantaneous secrecy rate maximization
by jointly optimizing the transmit power of the Tx and jammer while maintaining the energy harvesting requirement of the ER. The optimization problem can be mathematically formulated as:
\begin{subequations}
\label{problem}
\begin{align}
{\rm (P1):}
\max_{\{p_{n}, q_n\}}\quad\quad&\sum\limits_{n=1}^N R_{n}  \\
{\rm s.t.}\quad\quad& \eqref{total_P}-\eqref{Q}, \nonumber\\
\quad\quad& 0\leq p_n\leq \bar{p}, 0\leq q_n\leq \bar{q},  \forall n.
\end{align}
\end{subequations}

(P1) is non-convex since the rate expression \eqref{rate_express} is non-concave in the power variables. However, it is easy to verify that (P1) satisfies the so-called time-sharing condition, and thus (P1) has  zero duality-gap. This means that (P1) can be solved optimally by the Lagrange duality method. In next section, we apply the Lagrange duality method to solve (P1).

\section{Optimal Algorithm}

At first, the lagrangian of (P1) is expressed as
\begin{small}
\begin{align}
&\mathcal{L}(\{p_n\},\{q_n\},\lambda,\beta,\mu)
=\sum\limits_{n=1}^N R_n + \lambda\bigg(P-\sum\limits_{n=1}^Np_n\bigg)\label{LL}\\
 &+ \beta\bigg(\sum\limits_{n=1}^Np_nh_{J,n}-\sum\limits_{n=1}^Nq_n\bigg)
 + \mu\bigg(\sum\limits_{n=1}^N(p_nh_{E,n}+q_ng_{E,n})-\overline{Q}\bigg), \nonumber
\end{align}
\end{small}
where $\lambda$, $\beta$ and $\mu$ are the non-negative dual variables associated with the corresponding constraints \eqref{total_P}, \eqref{jammer_EH_C} and \eqref{Q}, respectively. Then,  the dual function $g(\lambda,\beta,\mu)$ of  (P1) is defined as
%
\begin{align}
\max_{0\leq p\leq \bar{p},0\leq q \leq \bar{q}}\mathcal{L}(\{p_n\},\{q_n\},\lambda,\beta,\mu). \label{Lagrange}
\end{align}
The dual problem is thus given by
\begin{align}
\min_{\lambda\geq 0,\beta\geq 0,\mu\geq 0} \quad g(\lambda,\beta,\mu) .
\end{align}

With a given set of $\{\lambda,\beta,\mu\}$, the maximization problem in \eqref{Lagrange} can be decomposed into $N$ parallel subproblems all having the same structure and each for one SC. By dropping the index $n$ for brevity, each subproblem is given by
\begin{align} \label{L_n}
L(p,q) = R - \lambda p + \beta (ph_{J}-q)
+ \mu (p h_{E}+q g_{E}).
\end{align}
In the next, we jointly optimize $p$ and $q$ in two different cases depending on $R\geq 0$ or $R = 0$ in  \eqref{rate_express}.

 When $q\geq A$:  We define $f_1(p,q)\triangleq\frac{\partial L}{\partial p}$, $f_2(p,q)\triangleq\frac{\partial L}{\partial q}$  and $\chi(\cdot)$ is the  real nonnegative  root of $f_1(p,q)=0$ and/or $f_2(p,q)=0$. Then the optimal solution $(p,q)$ in this case is given by the following proposition.
\begin{proposition} \label{solution_1}
The optimal solution of  (P1) with $q \geq A$ is given by

If $f_1(0,A)\leq f_1(0,\bar{q})\leq 0$,
\begin{eqnarray} \nonumber
 p=0,\quad q=\left\{
                    \begin{array}{ll}
                      A, & \hbox{if $-\beta+\mu g_{E}<0$} \\
                      \bar{q}, & \hbox{otherwise.}
                    \end{array}
                  \right.
\end{eqnarray}

If  $f_1(0,\bar{q}) \geq f_1(0,A) \geq 0$
\begin{align} \nonumber
\left\{
  \begin{array}{ll}
    p = \bar{p},\quad q=\chi(f_2(\bar{p},q)),      &\hbox{if $ f_1(\bar{p},A) \geq 0$} \\
    (p,q) = \chi(f_1(p,q),f_2(p,q)),     &\hbox{if $ f_1(\bar{p},\bar{q}) \leq 0$} \\
    (p,q) = \argmax\limits_{(p,q)\in \Upsilon_1} L(p,q),   &\hbox{otherwise,}
  \end{array}
\right.
\end{align}
where $\Upsilon_1$ is denoted as
\begin{small}
\begin{eqnarray} \nonumber
\Upsilon_1 = \left\{
               \begin{array}{ll}
                 (p,q)=\chi(f_1(p,q),f_2(p,q)), & \hbox{$A\leq q \leq \chi(f_1(\bar{p},q))$} \\
                 p = \bar{p},q = \chi(f_2(\bar{p},q)), & \hbox{$\chi(f_1(\bar{p},q)) \leq q \leq \bar{q}$.}
               \end{array}
             \right.
\end{eqnarray}
\end{small}

If $f_1(0,A) \leq 0 \leq f_1(0,\bar{q})$
\begin{align} \nonumber
\left\{
  \begin{array}{ll}
    (p,q) = \argmax\limits_{(p,q)\in \Upsilon_2} L(p,q),  \hbox{if $ f_1(\bar{p},\bar{q})  \leq 0$} \\
     (p,q) = \argmax\limits_{(p,q)\in \Upsilon_3} L(p,q),  \hbox{  otherwise,}
  \end{array}
\right.
\end{align}
where $\Upsilon_2$ and $\Upsilon_3$ are given by
\begin{footnotesize}
\begin{eqnarray} \nonumber
\Upsilon_2 = \left\{
               \begin{array}{ll}
                 p=0,q=\left\{
                         \begin{array}{ll}
                           A,  \hbox{$-\beta+\mu g_E<0$} \\
                           \chi(f_1(0,q)),  \hbox{otherwise.}
                         \end{array}
                       \right.
 & \hbox{$A\leq q \leq \chi(f_1(0,q))$} \\
                 (p,q)=\chi(f_1(p,q),f_2(p,q)),  & \hbox{$\chi(f_1(0,q))\leq q \leq \bar{q}$.}
               \end{array}
             \right.
\end{eqnarray}
\end{footnotesize}
\begin{footnotesize}
\begin{eqnarray} \nonumber
\Upsilon_3 = \left\{
               \begin{array}{ll}
                 p=0,q=\left\{
                         \begin{array}{ll}
                           A,  \hbox{$-\beta+\mu g_E<0$} \\
                           \chi(f_1(0,q)),  \hbox{otherwise.}
                         \end{array}
                       \right.
  \hbox{$A\leq q \leq \chi(f_1(0,q))$} \\
                 (p,q)=\chi(f_1(p,q),f_2(p,q)),  \hbox{$\chi(f_1(0,q))\leq q \leq \chi(f_1(\bar{p},q))$} \\
                 p = \bar{p},q=\chi(f_2(\bar{p},q)),  \quad \quad \quad \quad  \quad\quad\quad\hbox{$\chi(f_1(\bar{p},q))\leq q \leq \bar{q}.$}
               \end{array}
             \right.
\end{eqnarray}
\end{footnotesize}

\end{proposition}

\begin{proof}
Please refer to Appendix A.
\end{proof}

 When $q< A$:  $R=0$ in this case, thus \eqref{L_n} becomes a linear function of $p$ and $q$. Hence, the optimal solution  is given by
\begin{small}
\begin{align}\nonumber
p =
\begin{cases}
 0, \mbox{if}  - \lambda + \beta h_{J} + \mu h_{E} < 0,\\
\bar{p},\mbox{otherwise}.
\end{cases}
q =
\begin{cases}
 0, \mbox{if}  - \beta + \mu g_{E} < 0,\\
A,\mbox{otherwise}.
\end{cases}
\end{align}
\end{small}

Obtaining  the optimal $p$ and $q$ in each region, and  we can select the  $(p^*,q^*)$  which achieves the largest value of $L(p,q)$ in \eqref{L_n} as the optimal solution with given dual variables.
%

Finally, we update the dual variables using ellipsoid method due to the fact that $P-\sum_{n=1}^Np_n,  \sum_{n=1}^Np_nh_{J,n}-\sum_{n=1}^Nq_n$ and $ \sum_{n=1}^N\big(p_nh_{E,n}+q_nh_{E,n}\big)-\overline{Q} $ are the subgradients of $\lambda$, $\beta$ and $\mu$, respectively.

\section{Numerical Results}

We set up  $N=64$ SCs, the noise power $ \sigma^2=-60$ dBm  and the pass-loss exponent is  3. The peak power constraints $\bar{p}=\bar{q}=2P/N$.
The Tx, jammer, and IR are on  one straight line and the distance from the Tx to IR is  20 m. The jammer moves from the Tx to the IR, and the distance between the Tx and  jammer  is denoted as $d_1$.
%
Besides, we assume that the distance from the Tx to ER is 10 m with 30 degrees.  For comparison, we introduce three benchmark schemes: the jamming signal cannot be cancelled at both IR and ER
 (a near-optimal solution is obtained by block-coordinate descent method),
without jammer, and the equal power allocation (EPA) applied at both Tx and jammer.

\begin{figure}[t]
\begin{centering}
\centering
\includegraphics[scale=0.5]{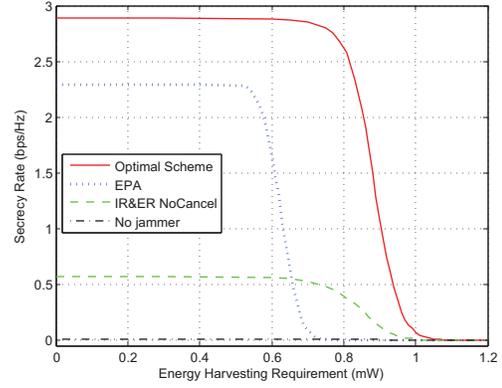}
\vspace{-0.1cm}
\caption{Secrecy rate versus required harvested power, with $P=30$ dBm.}\label{MSC_Q_R}
\end{centering}
\vspace{-0.1cm}
\end{figure}

Fig. \ref{MSC_Q_R} demonstrates the secrecy rate versus the harvested energy constraint $\overline{Q}$ with the total transmit power of Tx set as $P=30$ dBm and $d_1=10$ m.
First, for all schemes, the secrecy rate is observed to decrease with $\overline{Q}$. It is also observed that the proposed scheme outperforms the other three benchmark schemes significantly.
The EPA shows good performance with small energy requirement while becomes infeasible when $\overline{Q}$ becomes larger ($\overline{Q}>0.7$ mW).
Moreover, the scheme without  jammer  has the worst performance which has almost zero secrecy rate. This is because that the ER is located  nearly to the Tx and thus possesses  much better channel gains compared with the IR, thus the secrecy is unable to be guaranteed.

\begin{figure}[t]
\begin{centering}
\centering
\includegraphics[scale=0.5]{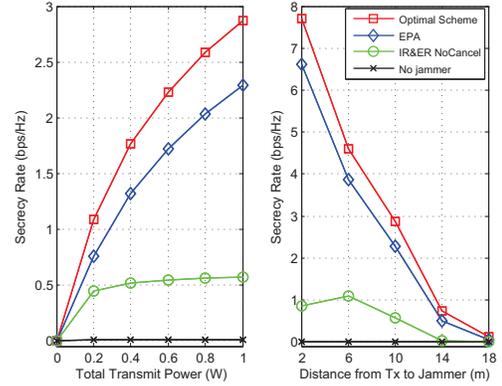}
\vspace{-0.1cm}
\caption{Secrecy rate versus the total transmit power and $d_1$, with $\overline{Q} = 100\mu$W.}\label{MSC_P_R}
\end{centering}
\vspace{-0.1cm}
\end{figure}

Fig. \ref{MSC_P_R} illustrates the secrecy information rate versus the the total transmit power $P$ with the harvested energy requirement $\overline{Q}=100$ $\mu$W and $d_1=10$ m. It also shows that the proposed optimal scheme achieves considerable gain compared with the other three benchmark schemes.
%
%
When   $d_1$ varies, the performance of the proposed scheme and EPA are observed to decrease with the distance $d_1$. This may be because that the jamming signals become weaker as the harvested energy of the jammer will decrease with a longer distance $d_1$ even if the channel gains from the jammer to the IR increases.
%
Besides,  the proposed scheme with AN cancellation at IR can achieve positive secrecy rate over a wide range compared with the scheme ``IR$\&$ER NoCancel", which also demonstrates the superiority of the proposed scheme.

\section{Conclusion}
This paper studied optimal power allocation in secure OFDM-based SWIPT system with the help of a wireless-powered friendly jammer.
%
%
The joint power allocation of the Tx and jammer were  optimized to maximize the  secrecy rate while satisfying the harvested power constraint.
We derived the optimal solution to solve the considered non-convex problem.
Finally, the superiority of the proposed scheme was  verified by the numerical results.

\appendices

\section{Proof of Proposition \ref{solution_1}}

As we define that
\begin{align}
f_1(p,q) \triangleq  & \frac{h_{I}}{\ln2(\sigma^2+ph_{I})} - \frac{h_{E}}{\ln2(\sigma^2+qg_{E}+ph_{E})} \nonumber\\
&- \lambda + \beta h_{J} + \mu h_{E},\label{piandaoshu_p_1} \\
f_2(p,q) \triangleq  &\frac{ph_{E}g_{E}}{\ln2(\sigma^2+qg_{E}+ph_{E})(\sigma^2+qg_{E})}
- \beta + \mu g_{E}. \label{piandaoshu_q_1}
\end{align}

We can easily prove that \eqref{piandaoshu_p_1} is a monotonic decreasing function with $p$, thus we have
$
f_1(\bar{p},q)\leq f_1(p,q) \leq f_1(0,q).
$
In addition, $f_1(p,q)$ is a monotonic increasing function of $q$ in $[A,\bar{q}]$. There are three cases about the sign of $f_1(0,q)$ which will be discussed in the following.

\textbf{Case I:} $f_1(0,A)\leq f_1(0,\bar{q})\leq 0$.
In this case,  $f_1(p,q) \leq f_1(0,q) \leq  f_1(0,\bar{q}) \leq 0$, thus $p=0$. Therefore, $f_2(0,q)$ becomes a linear function of $q$, thus the optimal $q$ can be given by
\begin{align}
q =
\begin{cases}
A, &\mbox{if} \quad  - \beta + \mu g_{E} < 0,\\
\bar{q},&\mbox{otherwise}.
\end{cases}
\end{align}

\textbf{Case II:} $f_1(0,\bar{q}) \geq f_1(0,A) \geq 0$.
In this case, $f_1(0,q)\geq f_1({0,A}) \geq 0$. In order to figure out the sign of $f_1(p,q)$, there are three subcases for $f_1(\bar{p},q)$  as follows:
\begin{itemize}
   \item \textbf{Case II-i:}   $f_1(\bar{p},\bar{q})\geq f_1(\bar{p},A) \geq 0$.
   In this subcase,  $f_1(p,q) \geq f_1(\bar{p},q) \geq f_1(\bar{p},A) \geq 0$, thus the optimal solution of $p$ is  $p=\bar{p}$. Then  the solution is $q = \chi(f_2(\bar{p},q))$  which can be found by the bisection over $[A,\bar{q}]$.

   \item \textbf{Case II-ii:}   $f_1(\bar{p},A)\leq f_1(\bar{p},\bar{q}) \leq 0$.
   In this subcase,  $f_1(\bar{p},q)\leq f_1(\bar{p},\bar{q}) \leq 0$. The solution is $\chi(f_1(p,q),f_2(p,q))$.  We can first eliminate $p$ in $f_2(p,q)$  and then find the optimal $q$ by numerical search over $[0,\bar{q}]$.

   \item \textbf{Case II-iii:}   $f_1(\bar{p},A) \leq 0 \leq f_1(\bar{p},\bar{q})$.
   In this subcase, $f_1(\bar{p},q)$ is not always positive or negative. We have $f_1(\bar{p},q)\geq 0$ when $q \geq \chi(f_1(\bar{p},q))$ and $f_1(\bar{p},q)<0$ otherwise. Therefore:

 \textbf{Region 1:}  $A \leq q < \chi(f_1(\bar{p},q))$.
  At this region, $f_1(\bar{p},q)<0$. Similar to Case II-ii, the optimal solution $(p,q)$ at this region is given by $\chi(f_1(p,q),f_2(p,q))$.

\textbf{Region 2:}   $\chi(f_1(\bar{p},q)) \leq q \leq \bar{q}$.
   At this region, $f_1(p,q) \geq f_1(\bar{p},q)\geq 0 $. Similar to Case II-i, we have $p=\bar{p}$ and $q = \chi(f_2(\bar{p},q))$  which can be found over $[\chi(f_1(\bar{p},q)),\bar{q}]$.

The optimal solution can be found via a simple search over $\Upsilon_1$ which defined as a set consisting the optimal solutions of the above two regions of case II-iii.

\end{itemize}

\textbf{Case III:} $f_1(0,A) \leq 0 \leq f_1(0,\bar{q})$.
In this case, $f_1(0,q)$ is not always positive or negative. We can easily get that $f_1(0,q)\leq 0$ when $q \leq \chi(f_1(0,q))$ and $f_1(0,q)>0$ otherwise.

\begin{itemize}
  \item \textbf{Case III-i:} $f_1(\bar{p},A) \leq f_1(\bar{p},\bar{q}) \leq 0$.
  In this subcase, $f_1(\bar{p},q) \leq f_1(\bar{p},\bar{q}) \leq 0$. According to the sign of $f_1(0,q)$, we obtain the optimal $p$ and $q$ in the following two regions.

   \textbf{Region 1:} $A \leq q \leq \chi(f_1(0,q))$.
    In this region, $f_1(p,q)  \leq   f_1(0,q) \leq 0$ and the optimal $p$ is given by $p = 0$. Similar to Case I, the optimal $q$ is given by
\begin{align} \label{q_n_b}
q =
\begin{cases}
A, &\mbox{if} \quad  - \beta + \mu g_{E} < 0,\\
\chi(f_1(0,q)),&\mbox{otherwise}.
\end{cases}
\end{align}
 \textbf{Region 2:} $\chi(f_1(0,q)) \leq q \leq \bar{q}$.
    In this region, $f_1(0,q)\geq 0$. Similar to Case II-ii, the optimal solution is given by $\chi(f_1(p,q),f_2(p,q))$.

  Therefore, we define $\Upsilon_2$  consisting the optimal solutions of the above two region and the optimal solution of this case can be given by a simple search over $\Upsilon_2$.

  \item \textbf{Case III-ii:} $f_1(\bar{p},A) \leq 0 \leq f_1(\bar{p},\bar{q})$.
  In this subcase, similar to $f_1(0,q)$, $f_1(\bar{p},q)$ is not always positive or negative. Since $f_1(\bar{p},q) \leq f_1(0,q)$, we can easily have $ \chi(f_1(\bar{p},q)) \geq \chi(f_1(0,q)) $. In the following, We jointly optimize $p$ and $q$ in the following three regions.

   \textbf{Region 1:} $A\leq q \leq \chi(f_1(0,q))$.
    In this region, $f_1(p,q)\leq f_1(0,q)\leq 0$, thus the optimal $p$ is given by $p=0$, and the optimal $q$ is given by \eqref{q_n_b}.

    \textbf{Region 2:}  $\chi(f_1(0,q)) \leq q \leq \chi(\bar{p},q))$.
    In this region, $f_1(0,q)\geq 0$ and $f_1(\bar{p},q) \leq 0$. Similar to Case II-ii,  we have $\chi(f_1(p,q),f_2(p,q))$ as the solution of this region.

 \textbf{Region 3:} $\chi(f_1(\bar{p},q)) \leq q \leq \bar{q}$.
    In this region, $f_1(p,q)\geq f_1(\bar{p},q)\geq 0$. Similar to Case II-i, we have $p = \bar{p}$ and $q = \chi(f_2(\bar{p},q))$ which can be obtained over $[\chi(f_1(\bar{p},q)), \bar{q}]$.

As a result, the optimal solution of Case III-ii can be found over $\Upsilon_3$ which consisting the three solutions given in the above regions.

\end{itemize}

\bibliography{secrecy_reference}
\end{document}